\documentclass{article}
\usepackage[a4paper, total={6in, 8in}]{geometry}
\usepackage[ansinew]{inputenc}
\usepackage[pdftex]{graphicx}
\usepackage[english]{babel}
\usepackage{amssymb,amsfonts}
\usepackage{amsmath}
\usepackage[T1]{fontenc}
\usepackage{wrapfig}
\RequirePackage{lineno}

\newcommand{\red}{\ensuremath{\mathcal{R}}}
\newcommand{\blue}{\ensuremath{\mathcal{B}}}

\DeclareMathOperator{\CH}{CH}

\renewcommand{\epsilon}{\ensuremath{\varepsilon}}

\graphicspath{{figures/}}

\usepackage {color}

\title{%
  Geodesic Order Types%
\thanks{
A preliminary version of this paper appeared in the proceedings of the 18th International Conference on Computing and Combinatorics (COCOON'12)\cite{akpv-got-12}. Other than this footnote, this is the extended version that was afterwards published in  in the special issue of Algorithmica~\cite{akpv-got-12j} containing the best papers of COCOON 2012 conference.}
}
\date{}
\author{
   Oswin Aichholzer\thanks{Institute of Software Technology, Graz University of Technology, Graz, Austria, 
	{\tt \{oaich,apilz,bvogt\}@ist.tugraz.at}} \and
   Matias~Korman\thanks{Universitat Polit\`ecnica de Catalunya, Barcelona, Spain 
	{\tt matias.korman@upc.edu}.} \and
   Alexander~Pilz\footnotemark[2] \and
   Birgit~Vogtenhuber\footnotemark[2]
}

\usepackage{amsthm}
\theoremstyle{plain}
\newtheorem{theorem}{Theorem}
\newtheorem{lemma}{Lemma}
\newtheorem{corollary}{Corollary}

\DeclareMathOperator{\cc}{cc}

\begin{document}

\sloppy
\maketitle

\begin{abstract}
The geodesic between two points $a$ and $b$ in the interior of a simple polygon~$P$ is the shortest polygonal path inside $P$ that connects $a$ to $b$.
It is thus the natural generalization of straight line segments on unconstrained point sets to polygonal environments.
In this paper we use this extension to generalize the concept of the order type of a set of points in the Euclidean plane to geodesic order types.
In particular, we show that, for any set $S$ of points and an ordered subset $\blue \subseteq S$ of at least four points, one can always construct a polygon $P$ such that the points of $\blue$ define the geodesic hull of~$S$ w.r.t.~$P$, in the specified order.
Moreover, we show that an abstract order type derived from the dual of the Pappus arrangement can be realized as a geodesic order type.
\end{abstract}

\section{Introduction}

Order types are one of the most fundamental combinatorial descriptions of sets
of points in the plane.  For each triple of points the order type encodes its
orientation and thus reflects most of the combinatorial properties of the given
set. We are interested in how much the order type of a point set changes when
the points lie inside a simple polygon, and the orientation of point triples is
given with respect to the geodesic paths connecting them. As depicted in
\figurename~\ref{fig:intro}, this orientation can change depending on the
polygon. In this paper we develop a generalization of point set order types to
the concept of geodesic order types.

\begin{figure}[htb]
  \centering
  \includegraphics[page=1]{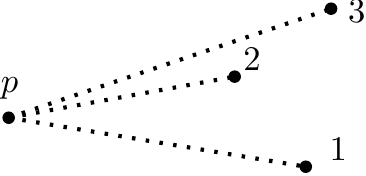}\hspace{2cm}
  \includegraphics[page=2]{intro_1_2}
  \caption{The radial order of shortest paths to points around a point $p$ can be 
different in unconstrained and geodesic settings.}
  \label{fig:intro}
\end{figure}

In set theory, order types impose an equivalence relation between ordered sets.  Two sets have the same \emph{order type} if
there is a bijection between them that is order
preserving~\cite[pp.~50--51]{set_theory}.  Goodman and
Pollack~\cite{goodman_pollack} extend this concept to finite, multidimensional
sets.  They define that two $d$-dimensional point sets $S_1$ and $S_2$ have the
same \emph{point set order type} when there exists a bijection $\sigma$ between
the sets such that each $(d+1)$-tuple in~$S_1$ has the same orientation (i.e.,
the side of the hyperplane defined by $p_1 \dots p_d$ on which the point
$p_{d+1}$ lies) as its corresponding tuple in $S_2$. It is also common to
consider two point sets to be of the same order type if all orientations are
inverted in the second set. In the plane, this means that for two sets of the
same order type, the ordered point triple $u$, $v$ and $w$ has the same
orientation (clockwise or counterclockwise) as $\sigma(u)$, $\sigma(v)$,
$\sigma(w)$. The infinitely many different point sets of a given cardinality
can therefore be partitioned into a finite collection of order types. 
The orientations of all triples of the point set determine for any two
given line segments whether they cross. Therefore, the order type
defines most of the combinatorial properties of a point
set.\footnote{It is common to regard the properties defined by
  orientations of triples as the combinatorial ones. There are further
  settings on point sets that can be seen as being combinatorial as
  well, e.g., asking whether the fourth point of a quadruple lies
  inside the circle defined by the first three ones
  (see~\cite{knuth}). Also, the circular sequence of a point set is a
  richer way of describing the combinatorics of point sets, totally
  implying the order type~\cite{circular_sequence}.} For example, its
convex hull, planarity of a given geometric graph (e.g., a
triangulation), its rectilinear crossing number, etc.\ only depend on
the order type. One might wonder whether every (consistent) assignment
of orientations to triples of an abstract set allows a realization as
a point set in the Euclidean plane. This is in general not true, not
even if the assignment fulfills axiomatic requirements. See Knuth's
monograph~\cite{knuth} for a detailed and self-contained discussion of
this topic.

Generalizing classic geometric results to geodesic environments is a
well-studied topic. For example, Toussaint~\cite{relativehull} generalized the
concept of convex hulls of point sets to geodesic environments.
Other topics like Voronoi Diagrams~\cite{geovoro},
Ham-sandwich Cuts~\cite{hamsan}, Linear Programming~\cite{geolin}, etc.\ have
also been covered. However, to the best of our knowledge, the concept of
geodesic order types has not been studied in the literature. Hence, it constitutes 
a natural and general extension to the above results.

The classic order type is often used to identify extremal settings for combinatorial problems on point sets. For example, finding sets which minimize the number of crossings in a complete geometric graph, or maximize the number of elements of a certain class of graphs (spanning trees, matchings, etc.) are typical applications. In a similar spirit, the geodesic order type might be used to investigate extremal properties in geodesic environments. Examples might be problems on pseudo-triangulations (the side chains of a pseudo-triangle are geodesics), guarding problems inside polygonal boundaries (there, shortest paths are geodesics), and related problems; see, e.g., \cite{pt_survey} for a recent survey on pseudo-triangulations.

\subsection{Preliminaries}
A closed polygonal path $P$ is called a {\em simple} polygon if no point of the
plane belongs to more than two edges of~$P$, and the only points that belong to
exactly two edges are the vertices of~$P$. A closed polygonal path~$Q$ is a 
{\em weakly simple} polygon if every pair of points on its boundary separates~$Q$
into two polygonal chains that have no proper crossings, and if the angles of a
complete traversal of the boundary of~$Q$ sum up to $2\pi$~\cite{relativehull}.
Observe that a simple polygon is a weakly simple polygon, but the reverse
is not true. Unless stated otherwise, all polygons are considered to be simple
herein.  We will follow the convention of including both, the interior and the
boundary of a polygon, when referring to it. The boundary of polygon $P$ will be
denoted by~$\partial P$.

The {\em geodesic} $\pi(s,t,P)$ between two points $s,t\in P$ in a simple
polygon $P$ is defined as the shortest path that connects $s$ to $t$, among all
the paths that stay within $P$. If~$P$ is clear from the context, we simply
write $\pi(s,t)$.  It is well known from earlier work that there always exists
a unique geodesic between any two points \cite{m-spaop-96}, even if $P$ is
weakly simple. Moreover, this geodesic is either a straight line segment or a
polygonal chain whose vertices (other than its endpoints) are reflex vertices
of $P$. Thus, we sometimes denote the geodesic as the sequence of these reflex
vertices traversed in the geodesic 
(i.e., $\pi(s,t)=\langle s=v_0, v_1, \ldots, v_k=t \rangle$).
When the geodesic~$\pi(s,t)$ is a segment, we say that $s$ {\em sees}~$t$ (and
vice versa).

For any fixed polygon $P$, a region $C\subseteq P$ is {\em geodesically convex} 
(also called \emph{relative convex}) if
for any two points \mbox{$p,q \in C$}, we have $\pi(p,q,P) \subseteq C$. The 
{\em geodesic hull} (\emph{relative convex hull}) $\CH_P(U)$ of a set $U$ is defined as the smallest (in terms of
inclusion) geodesically convex region $C$ that contains $U$. We will denote by
$\CH(U)$ the standard Euclidean convex hull. Whenever a point $p\in U$ is in the
boundary of $\CH_P(U)$, we say that $p$ is an {\em extreme} point  of $U$ 
(with respect to $P$). The set of all such extreme points is called the 
{\em extreme} set of $U$, and is denoted by $E_P(U)$. 

Although these definitions are valid for any subset $U$ of $P$, in this paper
we will only use them for a finite set of points $S = \{p_1, \dots, p_n\}$.
Further note that the geodesic hull is a weakly simple polygon;
see \figurename~\ref{fig_example}. From now on we assume that the points in
the union of~$S$ with the set $V$ of vertices of $P$ are in {\em strong general}
position. That is, there are no three collinear points, and, for any four
distinct points $p_1,p_2,p_3,p_4 \in S\cup V$, the line passing through $p_1$
and $p_2$ is not parallel to the line passing through $p_3$ and $p_4$. 

\begin{figure}[htb]
  \centering
  \includegraphics[scale=1]{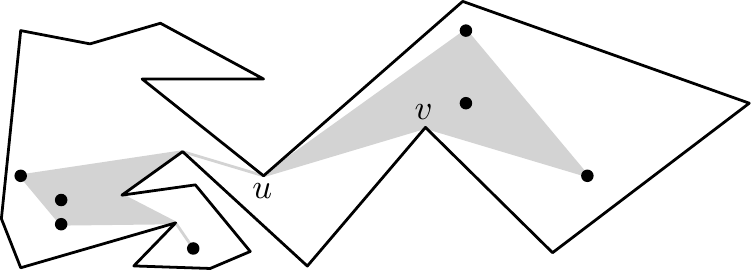}
  \caption{Seven points inside a polygon~$P$ and their geodesic hull (marked
	in gray). Observe that the boundary of the geodesic hull consists of the
	concatenation of the shortest paths connecting the extreme vertices of~$S$, in
	circular order. Further note that a vertex of the geodesic hull that stems from
	$P$ can be a convex and a reflex vertex of the geodesic hull at the same time
	(like vertex~$u$) or only a reflex vertex (like~$v$).} 
  \label{fig_example}
\end{figure}

\subsection{Orientations and Geodesics}
The concept of {\em clockwise order} of a triple of points $(p,q,r)$ naturally
extends to geodesic environments. Let $\pi(p, q) = \langle p=v_0,\ldots, v_k=q
\rangle$ and $\pi(p,r) = \langle p=u_0,\ldots, u_{k'}=r \rangle$ be the
geodesics  connecting $p$ with $q$ and $r$, respectively. Also, let $i>0$ be
the smallest index such that $v_i\neq u_i$. We say that $(p,q,r)$ are in
\emph{geodesic clockwise order} if $(v_{i-1},v_{i},u_i)$ are in (Euclidean)
clockwise order. It is easy to see that, due to the strong general-position
assumption, any triple is oriented either clockwise or counterclockwise in the
geodesic environment. We adopt the common phrasing, and say that $r$ is to the
right of $q$ (with respect to $p$) whenever $(p,q,r)$ are in geodesic clockwise
order (or that $r$ is to the left, otherwise). By definition, if $(p,q,r)$ are
in geodesic clockwise order, then for any $a<i \leq b,c$, the triple
$(v_a,v_b,u_{c})$ must also be in geodesic clockwise order.
Hence, this definition also accounts for the intuitive perception of ``left''
and ``right'' when traversing the geodesics.

Note that ``left'' and ``right'' differ between the geodesic and the
unconstrained setting, since we can use reflex vertices of the surrounding
polygon to ``reorder''
unconstrained point triples. An illustration is
shown in \figurename~\ref{fig:reorder}; in this example, the polygonal chain
crosses two edges of the triangle and the supporting line of the third one. In
general, this operation is not local, and might alter the order type of other
triples (more details of this operation will be given in
Section~\ref{sec:pappus}).

The orientation predicate can also be defined in terms of the
geodesic hull $\CH_P(\{p,q,r\})$. When traversing this hull counterclockwise, the
points appear in that order if and only if their geodesic orientation is
counterclockwise.

\begin{figure}[tb]
  \centering
  \includegraphics{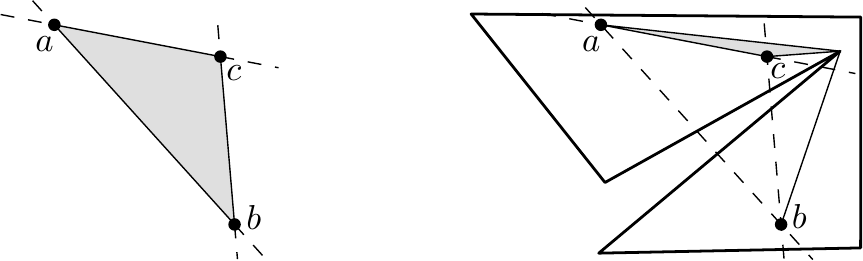}
  \caption{Reordering a triangle using a polygonal chain. The triple $(a,b,c)$
	is in (Euclidean) counterclockwise order. However, upon introducing the polygon
	(right figure) the same triple is now in (geodesic) clockwise order. }
  \label{fig:reorder} 
\end{figure}

\subsection{Contribution}
The triple orientation in geodesic environments extends the one in Euclidean
environments. Since the latter defines the order type of a point set, we obtain
a generalization of point set order types to \emph{geodesic order types}.  
It is easy to see that the order type of a fixed point set~$S$ can change with
different enclosing polygons. In particular, some points that appear in the
(Euclidean) convex hull may not be present in the geodesic hull and, vice
versa, some non-extreme points of~$S$ may appear on the geodesic hull. 

In this paper, we study the ways in which the set of extreme points of a given
set $S$ can change with the shape of the polygon. 
We show that any subset~$\blue$ of four or more points of~$S$ can become the
extreme set of~$S$ (i.e., there exists a polygon $P$ such that $E_P(S)=\blue$).
Moreover, we can make them appear in any predefined order along the boundary of
the geodesic hull. We also characterize when this property is fulfilled for
sets of size~$3$.
Finally, we show in Section~\ref{sec:pappus} that the abstract order types that
can be realized as geodesic order types are a proper superset of the abstract
order types realizable as Euclidean order types. Specifically, we show that the
non-realizable abstract order type derived from Pappus' Theorem via duality can
be realized as a point set inside a polygon. 

Our approach can also be seen as the class of inverse problems to the classic 
questions for geodesic environments, where the polygon is usually part of the input.

\section{Geodesic Hull versus Convex Hull}
\label{sec:geodesic_hull}

In this section, we study how much the geodesic hull of a given point set can
alter from the Euclidean convex hull. We partition $S$ into two sets of blue
and red points ($\blue$ and $\red$, respectively). A set $\blue$ is said to be
\emph{separable} from $\red$ if there exists a polygon with at most~$|\blue|$
convex vertices (i.e., a pseudo-$|\blue|$-gon) that contains all points of
$\red$ and no point of $\blue$ in its interior. From now on, we assume that the
set $S$ is fixed. Thus we omit writing ``from $\red$'' and simply refer to
$\blue$ as a separable point set.
The following theorem draws a nice connection between the separability of point
sets and their geodesic hull. 

\begin{figure}[htb]
  \centering
	\includegraphics[page=3]{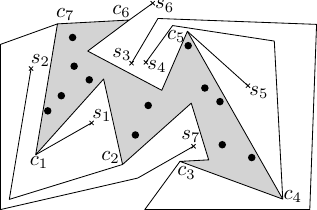}\hspace{2cm}
  \includegraphics[page=4]{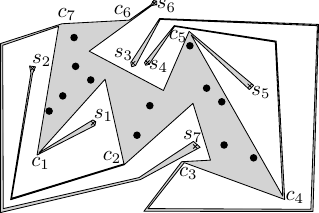}
  \caption{Illustration of the proof of Theorem \ref{theo_hull}: with a
	one-to-one correspondence between the convex vertices $c_1, \dots , c_7$ 
	of $P$ and the points $s_1, \dots , s_7$ of $\blue$, we can obtain a 
	weakly simple polygon $P'$ such that $E_{P'}(S)=\{s_1, \dots , s_7\}$ (left), 
	which then can be transformed to a polygon~(right).}
  \label{fig_separ}
\end{figure}
 
\begin{theorem}\label{theo_hull}
For any separable point set $\blue$ and any permutation $\sigma$ of $\blue$,
there exists a polygon~$P$ such that $E_P(S)=\blue$ and the clockwise ordering
of $\blue$ on the boundary of $\CH_P(S)$ is exactly $\sigma$.
\end{theorem}
\begin{proof}
Let $k=|\blue|$ and $P$ be a separating polygon of $\blue$. If $P$ has strictly
less than $k$ convex vertices, we introduce more by replacing any edge~$e$ by
two edges, adding a convex vertex arbitrary close to the center point of~$e$.
Thus, we assume that $P$ has $k$ convex vertices $c_1,\ldots, c_k$.

Let $s_1,\ldots, s_k$ be an arbitrary ordering of the vertices of $\blue$. For
all $i\leq k$, we connect point $s_i\in\blue$ to $c_i$ by a polygonal chain.
Observe that we can always do this in a way that no two chains cross.
Now let $P'$ be the union of $P$ and the polygonal chains; see
\figurename~\ref{fig_separ}~(left). The union of geodesics connecting $s_i$
with $s_{i+1}$ (and $s_k$ with $s_1$) exactly corresponds to the boundary
of~$P'$. Moreover, all points of $\red$ are in the interior of $P'$. Notice
that $P'$ is not a polygon, but a weakly simple polygon. 
As illustrated in \figurename~\ref{fig_separ}~(right), we obtain a polygon from $P'$ 
by transforming polygonal paths into narrow passages of width at 
most~$\epsilon$ (for a sufficiently small~$\epsilon$, and such that no blue 
point sees any other blue point).
\end{proof}

We now study the separability of a point set as a function of its
size. Surprisingly, the separability of the set $\blue$ does not
strongly depend on the set $\red$.

\begin{theorem}\label{theo_five_or_more}
Any set $\blue$ with cardinality $|\blue| \geq 5$ is separable with a polygon
with at most $2|\blue|-2$ vertices.
\end{theorem}

In order to prove the above theorem, we first consider some simpler cases and then show how to deal with larger point sets.

\begin{lemma}\label{lem_separ5}
Any set $\blue$ of five points is separable.
\end{lemma}

\begin{proof}
It is well-known that any set of five points contains a
convex quadrilateral $abcd$ that does not contain the fifth point~$e$. 
The supporting lines of the edges $ab$ and $cd$ cross due to the
strong general position assumption (analogously for the supporting
lines of $bc$ and $da$). These pairs of supporting lines define two
wedges that contain $abcd$, and at least one of them does not contain
the fifth point~$e$ (since their only region of intersection is the
quadrilateral). W.l.o.g., let this be the wedge defined by the
supporting lines of $ab$ and $cd$ and let $m$ be the crossing point of
its two supporting lines.  Further, assume that $m$ lies on
the ray from $a$ through $b$ and that the supporting line of $ab$
separates $e$ from $c$ and $d$; see
\figurename~\ref{fig_five_separable}~(left). 
We build two narrow polygonal spikes that 
contain the blue points and end sufficiently far away from the point
set. Each spike has a positive aperture angle at its unbounded end
and a sufficiently small aperture such that it does not 
contain any red point.  
The first spike starts on the line through $c$ and $d$ in a way that it contains
$c$, $d$, and $m$.
At~$m$, the spike bends towards $b$ and $a$ (with a
slightly positive aperture angle). 
The second spike 
contains~$e$, has its bisector parallel 
to the supporting line of $ab$, and 
is directed in the opposite direction of
the first spike; see again \figurename~\ref{fig_five_separable}~(left). 
Let $l$ and $l'$ be two lines which are parallel to the supporting line of $ab$
and slightly outside the convex hull of~$S$ (one line on each side). 
Since, by construction, the bisectors of
the (last parts of the) two spikes are parallel to $l$ and $l'$, the pair of
rays emanating from the end of each spike intersect $l$ and $l'$. These intersection points
become the end points of the spikes. Thus, they form a convex quadrilateral
containing all the points of~$S$, implying that the resulting polygon is a 
separating pseudo-5-gon; see \figurename~\ref{fig_five_separable}~(right).
\end{proof}

\begin{figure}[htb]
  \centering
  \hspace{0.2cm}\includegraphics[page=2]{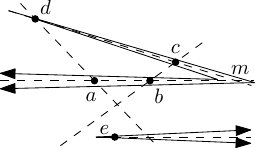}\hspace{2cm}
  \includegraphics[page=3]{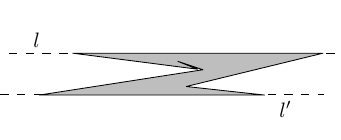}
  \caption{A set of five points is always separable. A narrow, bent
    spike can be built around the empty convex quadrilateral of the
    set. A second spike is chosen parallel to the first one and in
    opposite direction. Sufficiently far away, the spike end points
    span a quadrilateral around the whole set.}
  \label{fig_five_separable}
\end{figure}

Note that we can use the same construction when $\blue$ consists of four
points in convex position. In such a case, we do not need the second spike
(containing~$e$), and only place one convex vertex of the pseudo-4-gon on the
supporting line of $ab$ in the opposite direction of the spike. If the three
convex points on the convex hull of the pseudo-4-gon are chosen sufficiently
far away from the points of $\blue$, the pseudo-4-gon will always cover the red point set.
This implies Corollary~\ref{cor_4convex}.

\begin{corollary}\label{cor_4convex}
Any set $\blue$ of four points in convex position is separable.
\end{corollary}

\paragraph{Remark.} The separating polygon used in the above construction is
likely to have a ``bad aspect ratio'', in the sense that its horizontal dilation 
is far larger than the one of the convex hull of the point set.
While examples can be constructed where this cannot be avoided, we note that for subsets 
$\blue$ of cardinality $\geq 6$, we might obtain more elegant separating polygons using 
a different construction.
In essence, that approach removes pairs of points of $\blue$ with thin wedges and uses a 
large enclosing triangle; see \figurename~\ref{fig_separ3} for an example.
The complete construction requires some case analysis on the order type of the point set, 
and is thus omitted.

\begin{figure}[htb]
\centering
\includegraphics[page=2]{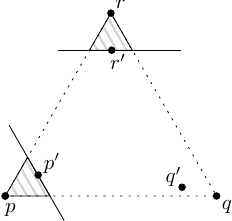}\hspace{2cm}
\includegraphics[page=3]{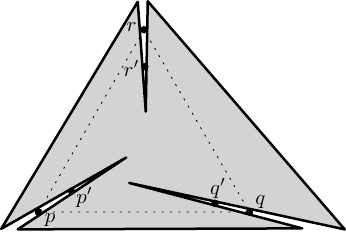}
\caption{One case for a construction to obtain a more ``nicely'' shaped polygon for $|\blue| = 6$.}
\label{fig_separ3}
\end{figure}

\begin{lemma} \label{lem_add}
For any set $\blue$ separable by a polygon $P$ and a point $q\not \in S$, the set $\blue \cup \{q\}$ is separable by a polygon $P'$ having at most two more vertices than $P$.
\end{lemma}

\begin{figure}[htb]
  \centering
  \includegraphics[page=3]{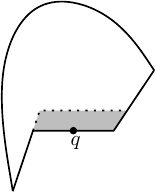}\hspace{1cm}
  \includegraphics[page=4]{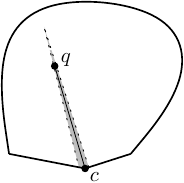}\hspace{1cm}
  \includegraphics[page=5]{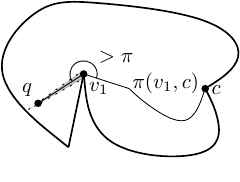}
  \caption{Proof of Lemma \ref{lem_add}. Regardless of whether $q$ is in
	$\partial P$, $q$ sees $c$ or a reflex vertex $v_1$, we can separate $\blue\cup
	\{p\}$. In all of the above cases, at most one convex vertex is added to $P$
	(as well as two edges).}
  \label{fig_add}
\end{figure}

\begin{proof}
Let $P$ be the polygon that separates $\blue$. Clearly, if $q\not\in P$, the
same polygon separates $\blue \cup \{q\}$. If $q\in\partial P$, it is easy to
do a small perturbation to $P$ such that $q$ is not contained in $P$ anymore;
see \figurename~\ref{fig_add} (left). 

Thus, we assume that $q$ is in the interior of $P$. Let $c$ be any convex
vertex of $P$.
If $q$ sees $c$, we remove $q$ from $P$ by adding a small spike emanating
from~$c$ towards~$q$; see \figurename~\ref{fig_add} (middle). In this operation
we replace a single convex vertex with two. Since we also increased the size of
the set by one, the separability invariant still holds. 

It remains to consider the case in which $q$ does not see $c$. Then,
the geodesic connecting $q$~and~$c$ is of the form $\langle
q=v_0,v_1,\ldots,v_k=c\rangle$ for some $k\geq 2$. By definition, $q$
sees $v_1$ and the interior angle $\angle{pv_1 v_2}$ is larger
than~$\pi$ (otherwise we could connect~$q$ towards~$v_2$ directly). In
this case we replace a reflex vertex with two vertices, but only one
of them will be convex; see \figurename~\ref{fig_add} (right).
\end{proof}

The class of polygons constructed in the proof of Lemma~\ref{lem_separ5}
will never have more than~$8$ vertices. Moreover, by Lemma \ref{lem_add}, each
additional point of $\blue$ will add at most $2$ additional vertices to the
separating polygon. In particular, we will always have a separating polygon $P$
whose number of edges is at most $2|\blue|-2$, which completes the proof of
Theorem~\ref{theo_five_or_more}. 

By definition, any point set of size $1$ or $2$ cannot be
separated (since we cannot construct a simple polygon with one or two convex
vertices). 
Hence, it remains to consider the cases in which $|\blue|\in \{3,4\}$.
Let $d$ be the shortest distance between any pair of blue points.
We say that a set $\red$ $\epsilon$-\emph{densely} covers $\blue$ (for any
$\epsilon >0$) if any wedge emanating from $p\in\blue$ and not containing any
point of $\red$ inside a circle with center $p$ and radius $d/2$ has an opening angle of at most $\epsilon$.
Observe that, if
$\red$ $\epsilon$-densely covers $\blue$, no point of $\blue$ can appear on the
boundary of $\CH(S)$.
Moreover, if $\epsilon\leq \pi/3$, any convex region that contains three or more blue points must contain a red point.
Showing that for any set $\red$ that $\epsilon$-densely covers $\blue$ 
(for some sufficiently small $\epsilon$), $\blue$ cannot be separated from $\red$, 
we obtain the following result.

\begin{theorem}\label{the_nonsepar34}
For any set $\blue$ of three points or four points in non-convex position, 
there exists a set $\red$ such that $\blue$ is not separable from $\red$. 
\end{theorem}

\begin{proof}
We claim that for any set $\red$ that $\epsilon$-densely covers $\blue$ (for 
some sufficiently small $\epsilon$), $\blue$~cannot be separated from $\red$. 
Assume that this is not true, and let $P$ be a separating polygon. 
Since the red set $\red$ is $\epsilon$-dense, every blue point 
must be inside a pocket of $P$
(where a pocket is a simple polygon defined by an edge of $\CH(P)$ and the sub-sequence of edges of $P$ between the two vertices of that edge).

If $|\blue| = 3$, the separating polygon has to be a pseudo-triangle,  
and every pocket is a side chain of $P$.
We define the \emph{aperture} of a side chain as the inner angle between the supporting lines of the first and the last edge of the side chain.
Since the red set is $\epsilon$-dense, at most two blue points can be separated
via the same side chain, and thus there must be at least two side chains 
enclosing blue points. Moreover, any such side chain has an aperture angle
of at most~$\epsilon$.
Consider the angular turn at a vertex of $P$, i.e., the signed angular change of direction when traversing the boundary of $P$. Recall that the sum of the angular turns a simple polygon is $2\pi$ and observe that due to the aperture of the pockets, the sum of the angular turns of the two pockets containing at least one blue point is $(-2 \pi) + 2 \epsilon$.
The angular turns of the three convex vertices can only add an amount strictly smaller than $3 \pi$ to that sum, which implies that we would need a fourth convex vertex to close the polygonal chain.

We now consider the case $|\blue| = 4$. Recall each pocket of $P$ is associated to a sub-sequence of edges that starts and ends at convex vertices of $P$. Moreover, convex vertices of $P$ in such a sub-sequence (other than the endpoints) correspond to reflex vertices of the pocket (and {\em vice versa}). Since $\CH(P)$ must have at least three vertices, $P$ can have at most one single pocket that is non-convex (and this situation can only happen when $\CH(P)$ is a triangle).

As $\blue$ is non-convex, any pocket containing all four blue points (and no red point) would need at least two convex vertices. This implies that there have to be at least two pockets containing blue points.
Let $k$ be the number of pockets of $P$ that contain blue points, and let $\beta_1,\ldots, \beta_k$ be the number of blue points contained in each pocket (in decreasing order).

Since $\sum_i \beta_i=4$ and $\beta_i\in\{1,\ldots, 3\}$, we distinguish between the following cases (depicted in Figure~\ref{fig_cases}):

\begin{figure}[htb]
\centering
\includegraphics[page=1]{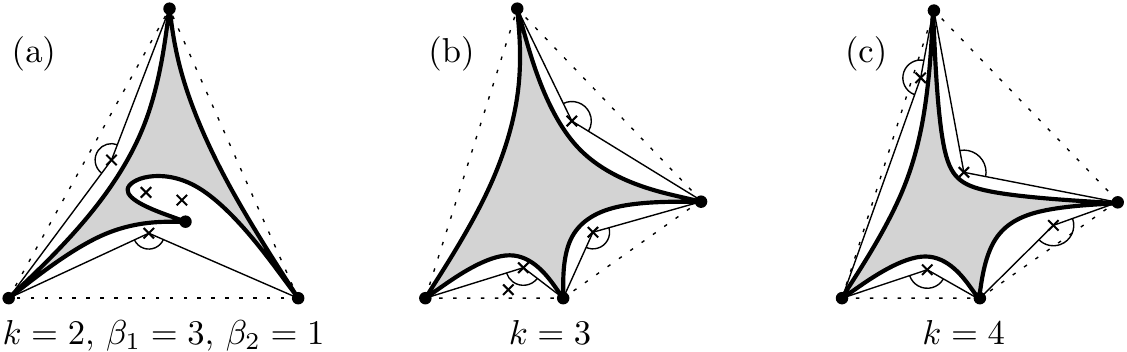}\\[2ex]
\includegraphics[page=2]{fig_cases}
\caption{Scheme of the possible configurations in which we can place four blue points in up to four pockets. In all cases, we obtain a contradiction, hence a separating polygon cannot exist.}
\label{fig_cases}
\end{figure}

\begin{description}
\item[Case $k=2$, $\beta_1=3$, $\beta_2=1$.] This case (depicted in Figure~\ref{fig_cases} (a)) is similar to the case $|\blue| = 3$. 
As a pocket containing three points must have a convex vertex, the convex hull of $P$ 
has three vertices and the two pockets containing the blue points must share a convex hull vertex (i.e., the side chains associated to each pocket share an endpoint).
One of the tree blue points sharing a pocket must see both convex hull vertices of that pocket, and therefore the aperture of that pocket is at most~$\epsilon$.
As before, the sum of angular turns is too small, and $P$ cannot be closed without introducing additional convex vertices.

\item[Case $k=2$, $\beta_1=\beta_2=2$.]
First consider the case in which there is a pocket $Q_1$ that is not convex.
If $Q_1$ does not contain a blue point or it contains a blue point that sees the convex hull vertices of $Q_1$, then we can argue in the same way as in the case where $|\blue| = 3$ (since the two pockets will be consecutive and each will have aperture at most $\epsilon$, see Figure~\ref{fig_cases} (d)).

Otherwise, no blue point sees both convex hull vertices of $Q_1$ (Figure~\ref{fig_cases} (e)).
In this case, we know that both blue points inside $Q_1$ see each other.
Let $Q_2$ be the second pocket containing blue points.
As in the previous cases, we know that the aperture of $Q_2$ is at most $\epsilon$.
Let $u$ and $v$ be the convex hull vertices defining $Q_2$ and let $w$ be the third vertex of the triangular convex hull of $P$.
W.l.o.g., $v$ and $w$ define the pocket $Q_1$. Let $b_1$ and $b_2$ be the two vertices of $\blue$ in the pocket $Q_2$ and let $b_3$ and $b_4$ be the ones in $Q_1$.
Further, let $\alpha>0$ be the smallest angle between $b_1 b_2$ and $b_3 b_4$.
We argue using bounds on the angles of the resulting polygon, see \figurename~\ref{fig_angles}.

\begin{figure}[b]
\centering
\includegraphics{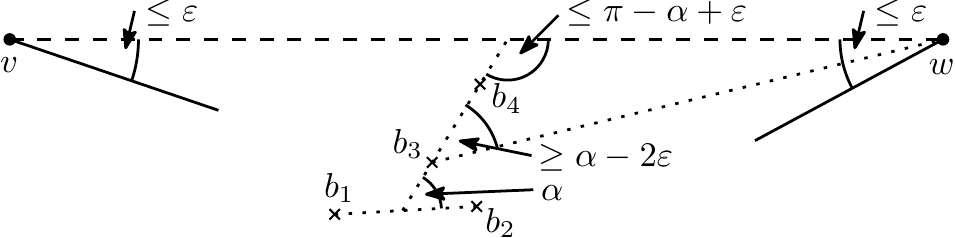}
\caption{Angles used in the proof that no pocket with a convex vertex contains two blue points. Note that alternatively, the supporting line of $b_1b_2$ could intersect $b_3b_4$.}
\label{fig_angles}
\end{figure}

Since the angular turns have to sum up to $2\pi$, we observe that the sum of the inner angles of all convex hull vertices is at most $\epsilon$.
The smallest angle between the convex hull edge $vw$ and $b_1 b_2$ is at most $\epsilon$,
since the aperture of $Q_2$ is at most $\epsilon$, and the inner angle at $v$ is at most $\epsilon$, but in the other direction.
This implies that the angle between $vw$ and $b_3 b_4$ is at least $\alpha - \epsilon$.
In particular, the supporting line of $b_3 b_4$ intersects the segment $vw$ if we choose $2\epsilon < \alpha$.
W.l.o.g., let $b_3$ and $b_4$ be arranged in a way that the ray from $b_3$ through $b_4$ intersects $vw$.
Barring symmetries, we have the situation shown in \figurename~\ref{fig_no_complex}.
Let $c$ be the convex vertex of $P$ in the pocket $Q_2$.
Observe that $c$ has to be in the same closed half-plane defined by the supporting line of $v b_4$ as the edge $vw$, as otherwise $b_4$ sees both $v$ and $w$ or the boundary of $P$ has another convex vertex between $v$ and $c$.
Since $c$ is separated from $b_3$ by the supporting line of $v b_4$,
the interior of the triangle defined by the supporting lines of $v b_4$, $b_4 b_3$, and $b_3 w$ is disjoint from $P$, and has an angle of at least $\alpha - 2\epsilon$ at $b_3$.
However, this contradicts the assumed $\epsilon$-density for a suitable choice of~$\epsilon$.

\begin{figure}[htb]
\centering
\includegraphics{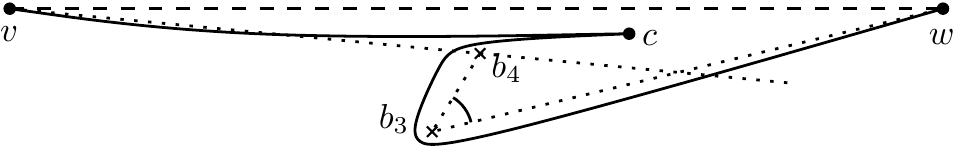}
\caption{No separating polygon can exist with two blue points in a pocket with a convex vertex.}
\label{fig_no_complex}
\end{figure}

It remains to consider the case in which all pockets are convex. By the non-convexity of $\blue$, pockets containing blue points cannot share an endpoint (Figure~\ref{fig_cases} (f)). However, in this case, the convex hull of the four pocket endpoints cannot contain all points of $\red$, implying  that $P$ cannot be a separating polygon.

\item[Case $k \geq 3$.]
Regardless of how many points are on 
the convex hull of $P$, notice that the pockets must share at least two endpoints (Figure~\ref{fig_cases} (b) and (c)),
and that all extreme vertices of $P$ must be pocket endpoints.
As the total aperture angle of the three pockets cannot be larger than $k\epsilon<\pi$, 
the polygon cannot be closed.\qedhere
\end{description}
\end{proof}

\begin{figure}[htb]
  \vspace{-2ex}
  \centering
  \includegraphics[page=2]{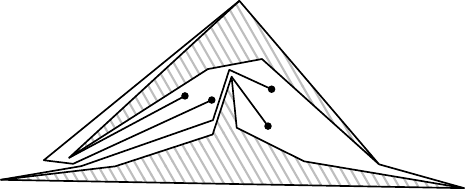}
  \caption{Two pseudo-triangles containing many red
    points such that the four blue points are not separable. However, they are the extreme vertices w.r.t.~some polygon.}
  \label{fig_non_separ_4}
\end{figure}

The example in \figurename~\ref{fig_non_separ_4} shows a point set
where the four blue points lie on the geodesic hull but are not separable.
This implies, in contrast to sets of larger cardinality, that for $|\blue|=4$,
the concepts of separability and geodesic hull are not equivalent.  
Thus, 
we switch back to the geodesic setting and consider 
the remaining cases $|\blue|\in\{3,4\}$.

\begin{lemma}\label{lem_geo_4}
For any set $S$, any set $\blue \subset S$ of four points, and any permutation $\sigma$ of $\blue$, 
there exists a polygon $P$ such that $E_P(S)=\blue$ and the clockwise ordering of $\blue$ on 
the boundary of $\CH_P(S)$ is exactly $\sigma$.
\end{lemma}
\begin{proof}
If the points of $\blue$ are in convex position, then the statement follows
directly from Corollary~\ref{cor_4convex} and Theorem~\ref{theo_hull}.  Thus,
assume that $\blue$ is not in convex position. Consider a line~$l_1$ spanned by
two of the extreme points of $\blue$, and a line $l_2$ that is parallel to
$l_1$ and passes through the third extreme point of $\blue$
(see~\figurename~\ref{fig_relative_hull_4}).  We construct two pseudo-triangles
$P_1$ and $P_2$, each with four edges, with the following properties: (1)~$P_1$
has a convex and a reflex vertex on $l_1$, such that the reflex vertex is
between the convex vertex and both blue points on $l_1$.  (2)~Accordingly,
$P_2$ has a convex and a reflex vertex on $l_2$, such that the reflex
vertex is between the convex vertex and the blue vertex on $l_2$.  (3)~Both,
$P_1$ and $P_2$, have a vertex between~$l_1$ and~$l_2$, and the edges
connecting the convex point on~$l_1$  ($l_2$) to these vertices are parallel.
(4)~The non-extreme point of $\blue$ lies between $P_1$ and $P_2$.  (5)~All red
points lie inside $P_1$ or $P_2$.  Note that these properties can always be
fulfilled, as the convex points of the pseudo-triangles can be far away, and
thus the reflex angles can be made arbitrarily small and the area covered by
the pseudo-triangles can be arbitrarily ``thick''.

As indicated in \figurename~\ref{fig_relative_hull_4}, we can merge the two
pseudo-triangles to form a polygon by adding a narrow passage from a convex
vertex of $P_1$ to a convex vertex of $P_2$.  To obtain our final polygon $P$
with $E_P(S)=\blue$ in the desired order, we proceed like in the proof of
Theorem~\ref{theo_hull}, connecting the blue points to the four convex vertices
of $P_1$ and $P_2$ that were not used for the passage between $P_1$ and $P_2$.
\end{proof}

\begin{figure}[htb]
  \centering
  \includegraphics[page=4]{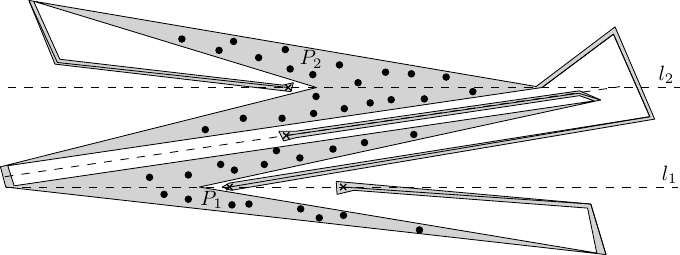}
  \caption{Construction for a polygon $P$ with $E_P(S)=\blue$ based on two
	pseudo-triangles that contain all red points (depicted with dots) and none of
	the blue points (drawn as crosses).}
  \label{fig_relative_hull_4}
\end{figure}

If we combine this result with Theorems~\ref{theo_hull}
and~\ref{theo_five_or_more} we obtain the following statement.

\begin{theorem}
\label{thm:equival3}
For any set $S$, any set $\blue \subset S$ of at least four points, and any permutation~$\sigma$ of~$\blue$, 
there exists a polygon $P$ such that $E_P(S)=\blue$ and the clockwise ordering of $\blue$ on 
$E_P(S)$ is exactly $\sigma$.
\end{theorem}

We conclude this section by studying what happens when the set $\blue$ has cardinality three.

\begin{theorem}\label{the_separ3}
Let $\blue \subset S$ be a set with $|\blue|=3$ such that $\blue$ spans the geodesic hull of~$S$ 
for some polygon $P$. Then $\blue$ is separable.
\end{theorem}
\begin{proof}
Recall that the geodesic hull of~$S$ is a weakly simple polygon 
which has all points of $\blue$ on its boundary, and contains all 
points of $S \setminus \blue$ in its interior.
Moreover, a vertex $v$ of the geodesic hull can only be convex if 
(1)~\mbox{$v \in \blue$}, or
(2)~$v$~is part of some weakly simple polygonal chain and thus 
coincides with a reflex vertex of the geodesic hull.
Thus, as $|\blue|=3$, the geodesic hull must consist of a
pseudo-triangle $\Delta$, possibly with polygonal chains attached to
the convex vertices of $\Delta$, where each blue vertex corresponds to
one convex vertex of $\Delta$; see \figurename~\ref{fig_geodesic_3}.
By slightly shrinking $\Delta$, we obtain a pseudo-triangle~$\Delta'$
still having all points of $S \setminus \blue$ in its interior that
leaves all points of $\blue$ outside. Thus $\Delta'$ is a separating
polygon for $\blue$.
\end{proof}

\begin{figure}[htb]
\centering
\includegraphics[page=2]{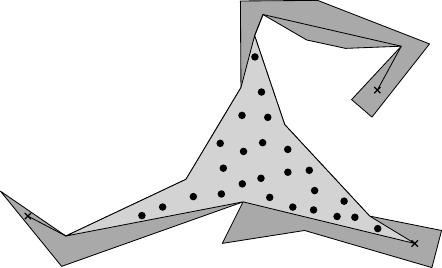}
\caption{A set $\blue \subset S$ with $|\blue|=3$, and a polygon $P$
  (dark shaded) with $E_P(S)=\blue$ (depicted $\times$). The geodesic hull is
  drawn light shaded.}
\label{fig_geodesic_3}
\end{figure}

\begin{table}[htb]
\centering
{\small
\begin{tabular}{r||l|l|l}
$|\blue|$ & Pushable & & Separable  \\
\hline
\hline
$\leq 2$  &  never (Def.)  & $\Leftrightarrow$ &  never (Def.)  \\
\hline
$3$  & not always  &  $\Leftrightarrow$~(Thm.~\ref{theo_hull} and \ref{thm:equival3}) &  not always~(Thm.~\ref{the_nonsepar34})  \\
\hline
$4$  &  always~(Thm.~\ref{lem_geo_4})  & $\Leftarrow$~(Thm.~\ref{theo_hull}) &  convex position: always~(Cor.~\ref{cor_4convex}) \\
     &       &     &  non-convex: not always~(Thm.~\ref{the_nonsepar34}) \\
\hline
$\geq 5$  &  always~(Thm.~\ref{thm:equival3})  & $\Leftrightarrow$ &  always~(Thm.~\ref{theo_five_or_more})  \\
\hline
\end{tabular}
}
\caption{Overview of results and relationship between pushable and separable.}
\label{tab:overview}
\end{table}

Together with Theorem~\ref{the_nonsepar34} the above result implies
that there exist point sets $S$ with $|\blue|=3$ such that $\blue$ can
not be used to define the geodesic hull of~$S$. This is in contrast to
the fact that for any set with $|\blue|\geq4$ this is always possible.
Table~\ref{tab:overview} gives an overview of the obtained results and
also shows the relation between a set being 'pushable' (meaning that
there is a polygon such that $\blue$ is on the geodesic hull) and
'separable' for different cardinalities of $\blue$.

\section{Realizing the Non-Pappus Arrangement}
\label{sec:pappus}
By duality, every set of points in the $d$-dimensional Euclidean space
corresponds to an arrangement of hyperplanes in the same space (see 
e.g.~\cite{edelsbrunner} for details on this mapping). This dual is
incidence and order preserving. When traversing a line $u^*$ in the
plane, the order in which the lines $v^*$ and $w^*$ are crossed gives
the orientation of the corresponding point triple $u, v, w$ in the
primal setting \cite{ordered_duality}. Hence, the crossings in the
line arrangement determine the order type of the corresponding point
set. An \emph{arrangement of pseudo-lines} is a set of simple curves
such that each pair has exactly one point in common, and at this point
the pair crosses. The crossings in the pseudo-line arrangement define
an \emph{abstract order type}. Obviously, if we can stretch the curves
to straight lines without changing the order of all crossings, we
obtain a realization of the order type defined by the crossings. This
has been used in the exhaustive enumeration of point set order
types~\cite{order_type_db}. However, for sets of size 9 or more, it is
known that there exist non-realizable abstract order types (i.e.,
pseudo-line arrangements that are non-stretchable). The example for 9
pseudo-lines is based on the well-known Pappus'
Theorem~\cite{gruenbaum,richter}.

Using the axiomatic system of~\cite[p.~4]{knuth}, one can show that geodesic order types are
in fact a subset of abstract order types, i.e., of those that are
defined by pseudo-line arrangements. Let the predicate $\cc(u,v,w)$ be
true whenever the point triple $(u,v,w)$ is oriented
counterclockwise. We already observed that $\cc(u,v,w) \Rightarrow
\cc(v,w,u)$, $\cc(u,v,w) \Rightarrow \neg \cc(u,w,v)$, and $\cc(u,v,w) \lor
\cc(u,w,v)$. Note that the latter holds since we require all points to
be strictly inside the surrounding polygon.  What remains to show is
that 
\[\cc(x,u,v) \land \cc(x,v,w) \land \cc(x,w,u) \Rightarrow \cc(u,v,w) \quad \mbox{and}\]
\[\cc(a,b,u) \land \cc(a,b,v) \land \cc(a,b,w) \land \cc(a,u,v) \land \cc(a,v,w) \Rightarrow \cc(a,u,w).\]
In other words, if $x$ is left of $\pi(u,v), \pi(v,w),$ and $\pi(w,u)$ then
$\CH_P(\{u,v,w\})$ is given by the sequence $\langle u,v,w \rangle$, and the points to the
left of $\pi(a,b)$ are in transitive radial order around $a$.
For the first of these statements, observe that since
$x$ cannot be on the geodesic hull of the four points, it is inside
the pseudo-triangular region of the hull. Hence, it is easy to see that
the implication is analogous to the Euclidean setting.  For the second
implication, consider the geodesics from $a$ to $u$, $v$, and $w$. If
they split at $a$, transitivity follows from the analogy to the
Euclidean setting. The same is the case if they split at the same
vertex $r$, as~$r$ is reflex.  If, say, $u$ splits first (the other
case is symmetric), it is clear that the orientation of $(a,u,v)$ is
the same as of $(a,u,w)$. It follows that all parts of the axiomatic system are
fulfilled, and therefore all geodesic order types are realizations of abstract order
types (cf.~\cite[pp.~23--35]{knuth}).

\subsection{The Arrangement}
\label{sec:arrangement}
\begin{figure}[htb]
\centering
\includegraphics[page=2]{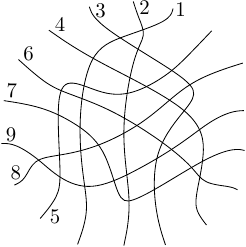}\hspace{2cm}
\includegraphics[page=3]{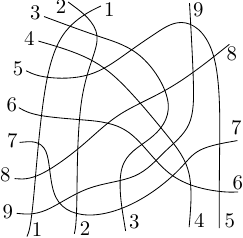}
\caption[A pseudo-line arrangement derived from Pappus' Theorem.]{A
  non-stretchable pseudo-line arrangement derived from Pappus'
  Theorem, adapted from~\cite[Fig.~5.3.2]{handbook} (left).  The
  transformed arrangement, having all lines crossing line 1 first
  (right).}
\label{fig:pappus_arr}
\end{figure}

The non-stretchable arrangement whose abstract order type we realize
in the geodesic setting is an adaption from the one shown in
\cite[p.~107]{handbook}; see
\figurename~\ref{fig:pappus_arr}~(left). It is well-known that this
pseudo-line arrangement cannot be stretched and thus the corresponding
abstract order type cannot be realized by a point set.
From the
correspondence between a straight line in the Euclidean plane to a
great circle in the sphere model of the projective plane, it is easy
to see that an arrangement is stretchable in the real plane if and only
if it is stretchable in the projective plane, provided that no pseudo-line in the projective plane coincides with the line at infinity.
We can therefore apply projective transformations to the arrangement
without affecting its realizability. In this way we transform the
arrangement of~\cite{handbook} to the \emph{standard labeling}; see
\figurename~\ref{fig:pappus_arr}~(right) for the resulting
drawing.
Roughly speaking, the crossings of a pseudo-line that happen before
the crossing with $l_1$ are ``moved'' to the other side. Namely, these
are the crossing of $l_9$ with $l_8$ and the crossings of $l_5$ with
$l_9, l_8, l_7$, and $l_6$, in the given order. We do so in order to
make all pseudo-lines cross pseudo-line $l_1$ before any other. In the
primal, this corresponds to $p_1$ being on the convex hull boundary
and points $p_2, \ldots, p_9$ being sorted clockwise around it. Note
that this kind of projective transformation actually preserves the
order type. Table~\ref{tab:pappus_arr_triples} 
shows all 
triples with ascending indices that have counterclockwise orientation
(which easily allows obtaining the orientation of all triples). For example, the
entry ``278'' indicates that pseudo-line $l_2$ crosses $l_8$ before $l_7$,
inducing counterclockwise orientation of the point triple $p_2 p_7
p_8$ in the primal.

\begin{table}
\centering
 \begin{tabular}{|c|c|c|c|c|}
\hline
278 & 345 & 467 & 567 & 678 \\
279 & 348 & 468 & 568 & 679 \\
    & 368 & 469 & 569 &     \\
    & 378 & 478 & 578 & \\
    & 379 & 479 & 579 & \\
    &     &     & 589 & \\
\hline
\end{tabular}
\caption{All ascending counterclockwise point triples derived from
  the arrangement.}
\label{tab:pappus_arr_triples}
\end{table}

\begin{figure}[htb]
\centering
\includegraphics[page=2]{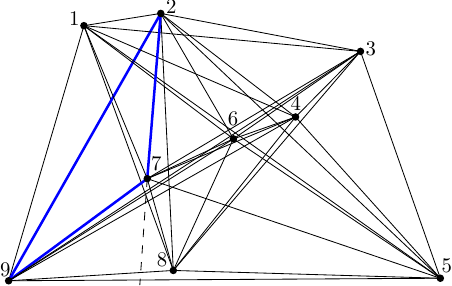}
\caption{A point set that ``almost'' realizes the unrealizable
  arrangement. The point triple
  spanning the thick blue triangle $\Delta p_2 p_7 p_9$ is the one for which the
  orientation is wrong.}
\label{fig:onered1}
\end{figure}

\begin{figure}[htb]
\centering

\includegraphics[page=2]{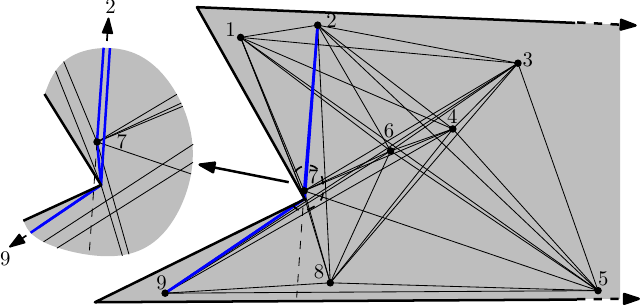}
\caption{A geodesic realization of the arrangement
  (right). The shortest paths between the points are geodesics in the
  interior of the polygon (gray). The region of interest is shown in
  detail in the middle. The polygon closes with a convex vertex far on
  the right side, as indicated.}
\label{fig:onered2}
\end{figure}

\subsection{The Realization}
\label{sec:realization}
Consider the point set $S = \{ p_1, \dots, p_9 \}$ shown in
\figurename~\ref{fig:onered1}. The only triples whose
orientations do not match those indicated by
\figurename~\ref{fig:pappus_arr} are the permutations of $p_2, p_7$
and $p_9$. Equivalently, one can say that the triangle defined by the
three points is the only one that has the wrong orientation among all
triangular subgraphs of the complete graph of~$S$. This triangle is
shown with thick (blue) edges.

We already discussed how reflex vertices of a surrounding polygon can
change the orientation of a triple.  The problem with this tool is
that the polygonal chain is likely to reorder other triangles as well.
In the point set shown in \figurename~\ref{fig:onered1}, this
tool can, however, be applied. We create a polygon $P$ that contains
$S$. The result of the construction is shown in
\figurename~\ref{fig:onered2}.

We cross four edges during this operation. Note that the geodesics
$\pi(p_1, p_9, P)$ and $\pi(p_1,p_8,P)$ are now no longer line
segments, still the order defined by their end vertices has not
changed. The triple $p_2, p_7, p_9$, however, is now oriented
counterclockwise, as demanded by the abstract order type. By checking
all the point triples, the reader can verify that this geodesic order
type indeed realizes the abstract order type of the non-Pappus
arrangement.

\begin{theorem}
  There exists a point set $S$ and a polygon whose geodesic order type
  realizes an abstract order type that is not realizable as a point
  set in the plane.
\end{theorem}

We note that our construction is minimal; that is, there cannot exist
a point set of nine points and a polygon of fewer vertices (than the
one given in \figurename~\ref{fig:onered2}) that realize the non-Pappus
arrangement.

There are 13 non-stretchable pseudo-line arrangements of 9 lines;
all these arrangements correspond to the same arrangement in the projective plane, i.e., the non-Pappus arrangement~\cite{richter}.
As already mentioned, the sphere model of the projective plane shows that a pseudo-line arrangement in the Euclidean plane is stretchable if and only if the corresponding arrangement in the projective plane is stretchable.
We found one realization for one abstract order type of the non-Pappus arrangement, however, we do not know whether the remaining 12 non-realizable abstract order types are realizable as a geodesic order type as well.

\section{Conclusion}
In this paper, we made a first step into generalizing the concept of
point set order types to geodesic order types. For a selection of four
or more points out of a set $S$, we showed how to construct a polygon
such that exactly these vertices are on the geodesic hull of~$S$, in
any order desired.  To the contrary, this is not always possible for
three points.
We further showed an example of an abstract order type that is not
realizable in the Euclidean plane, but is realizable in geodesic
environments.

Several interesting questions rise from our investigations. Which
bounds on the number of vertices in the polygon that forces the
desired geodesic hull can we derive? What is the complexity of
minimizing the number of vertices? Even though we showed the
realizability of the abstract order type derived from Pappus' Theorem,
we have no general tools to realize order types inside polygons. Can every abstract order type (which is non-realizable in the Euclidean plane) be realized as a geodesic order type? And which of them can be realized in a given polygon?
If not all of them can be realized, does realizability of an order type imply realizability of all abstract order types that correspond to the same pseudo-line arrangement in the projective plane?

\section{Acknowledgements}
The authors would like to thank Prosenjit Bose, Stefan Langerman, and Pat Morin for the introduction of the topic and for several fruitful discussions on it.

\bibliographystyle{splncs03}
\bibliography{bibliography}

\end{document}